\documentclass[11pt,letterpaper]{amsart} 

\usepackage{amsmath,amsfonts,amsthm,amssymb,stmaryrd,bm,cite,enumerate}
\theoremstyle{plain}

\numberwithin{equation}{section}

\newtheorem{thm}{Theorem}[section]
\newtheorem{lem}[thm]{Lemma}
\newtheorem{cor}[thm]{Corollary}

\theoremstyle{definition}
\newtheorem{example}{Example}

\theoremstyle{definition}

\allowdisplaybreaks  

\newcounter{cond}

\newcommand{\real}{{\mathbb R}}

\newcommand{\trace}{\mathrm{tr\,}}
\newcommand{\rmin}{\mathrm{In\,}}
\newcommand{\rmob}{\mathrm{Ob\,}}
\newcommand{\rmsub}{\mathrm{Sub}}
\newcommand{\ityes}{\textit{yes}}      
\newcommand{\itno}{\textit{no}}

\newcommand{\escript}{\mathcal{E}}
\newcommand{\hscript}{\mathcal{H}}
\newcommand{\iscript}{\mathcal{I}}
\newcommand{\jscript}{\mathcal{J}}
\newcommand{\kscript}{\mathcal{K}}
\newcommand{\lscript}{\mathcal{L}}
\newcommand{\oscript}{\mathcal{O}}
\newcommand{\sscript}{\mathcal{S}}

\newcommand{\iscriptbar}{\overline{\iscript}}

\newcommand{\iscripthat}{\widehat{\iscript}}
\newcommand{\jscripthat}{\widehat{\jscript}}
\newcommand{\kscripthat}{\widehat{\kscript}}

\newcommand{\ab}[1]{\left|#1\right|}
\newcommand{\brac}[1]{\left\{#1\right\}}
\newcommand{\paren}[1]{\left(#1\right)}
\newcommand{\sqbrac}[1]{\left[#1\right]}

\begin{document}

\title{iNTERVAL EFFECT ALGEBRAS AND HOLEVO INSTRUMENTS}

\author{Stan Gudder}
\address{Department of Mathematics, 
University of Denver, Denver, Colorado 80208}
\email{sgudder@du.edu}
\date{}
\maketitle

\begin{abstract}
This article begins with a study of convex effect-state spaces. We point out that such spaces are equivalent to interval effect algebras that generate an ordered linear space and possess an order-determining set of states. We then discuss operations and instruments on interval effect algebras under the assumption of an unrestrictive condition. Effects measured by operations and sequential products of effects relative to operations are considered. Observables are introduced and coexistence of effects are discussed. We also present properties of sequential products of observables and conditioning of observables related to instruments. The final section is devoted to Holevo instruments. Pure and mixed Holevo operations are defined and extended to instruments. The Holevo sequential product of two observables is defined and the marginals of these products are computed. We define the commutant of two effects and derive its properties. Examples are given that illustrate the properties of  previously presented concepts.
\end{abstract}

\section{Introduction}  
This introduction presents simplified definitions of various concepts. Complete and rigorous definitions as well as discussions of their meanings are left for later sections. The aim here is to give the reader an overview of the material presented subsequently.

The present article begins with a study of effect-state spaces $(\escript ,\sscript , F)$ where $\escript$ is the set of effects and $\sscript$ is the set of states for a physical system. An effect describes a simple \ityes-\itno\ (true-false) experiment and a state corresponds to the initial condition of a physical system. If an effect $a$ is measured and the result is \ityes\ (true), we say that $a$ \textit{occurs}. If the result is \itno\ (false), we say that $a$ \textit{does not occur}. The function $F\colon\escript\times\sscript\to\sqbrac{0,1}\subseteq\real$ gives the probability $F(a,s)$ that $a$ occurs when the system is in the state $s$. The axioms for
$(\escript ,\sscript ,F)$ imply that there is a null effect $a$ that never occurs and a unit effect $u$ that always occurs. Moreover, if $F(a,s)+F(b ,s)\le 1$ for every
$s\in\sscript$, there exists a unique $c\in\escript$, that we denote by $c=a+b$ such that $F(c,s)=F(a,s)+F(b,s)$ for all $s\in\sscript$. We say that $(\escript ,\sscript ,F)$ is a \textit{convex} effect-state space of for all $a\in\escript$ and $\lambda\in\sqbrac{0,1}\subseteq\real$ there exists an effect $\lambda a\in\escript$ such that
$s(\lambda a)=\lambda s(a)$ for every $s\in\sscript$. If $s(a)\le s(b)$ for all $s\in\sscript$, we write $a\le b$ and form the \textit{order-interval}
\begin{equation*}
\escript =\sqbrac{\theta ,u}=\brac{a\in\sscript\colon 0\le a\le u}
\end{equation*}
We point out in Section 2 that $(\escript ,\sscript ,F)$ is a convex effect-state space if and only if $\paren{\sqbrac{\theta ,u},\theta ,u,+}$ is an interval effect algebra with an order-determining set of states $\sscript$ and $\sqbrac{\theta ,u}$ generates an ordered linear space $V$.

Section 3 discusses instruments on $\sqbrac{\theta ,u}$ when it is assumed that $\sqbrac{\theta ,u}$ satisfies an unrestrictive condition. If 
$\escript _1=\sqbrac{\theta _1,u_2}$, $\escript _2=\sqbrac{\theta _2,u_2}$ are interval effect algebras, we study the set of operations $\oscript (\escript _1,\escript _2)$ from $\escript _1$ to $\escript _2$. Denoting the set of states on $\escript$ by $\sscript (\escript )$, a function $f\colon\escript\to\sqbrac{0,1}$ is a \textit{substate} if
$f(a)=\lambda s(a)$ for all $a\in\escript$, where $\lambda\in\sqbrac{0,1}$ and $s\in\sscript (\escript )$. We denote the set of substates by
$\rmsub\paren{\sscript (\escript )}$. An affine map $\iscript\colon\sscript (\escript _1)\to\rmsub\paren{\sscript (\escript _2)}$ is called an \textit{operation} from $\escript _1$ to $\escript _2$ and an operation $\iscript\colon\sscript (\escript _1)\to\sscript (\escript _2)$ is called a \textit{channel}. An \textit{instrument}
$\iscript =\brac{\iscript _x\colon x\in\Omega _\iscript}$, where $\iscript _x\in\oscript (\escript _1,\escript _2)$ and $\sum\brac{\iscript _x\colon x\in\Omega _\iscript}$ is a channel, is an operation-valued measure with \textit{outcome space} $\Omega _\iscript$. An instrument describes an experiment with outcomes in $\Omega _\iscript$ and when the outcome $x\in\Omega _\iscript$ occurs, the initial state $s\in\sscript (\escript _1)$ is updated to the (unnormalized) state
$\iscript _x(s)\in\rmsub\paren{\sscript (\escript _2)}$. We denote the set of instruments from $\escript _1$ to $\escript _2$ by $\rmin (\escript _1,\escript _2)$.

In Section 3 we show that if $\iscript\in\oscript (\escript _1,\escript _2)$, then there exists a unique affine dual map $\iscript ^*\colon\escript _2\to\escript _1$ that satisfies
$s\sqbrac{\iscript ^*(a)}=\iscript (s)(a)$ for all $s\in\sscript (\escript _1)$, $a\in\escript _2$. Also, $\iscript$ is a channel if and only if $\iscript ^*(u_2)=u_1$. We say that
$\iscript\in\oscript (\escript _1,\escript 2)$ \textit{measures} the effect $\iscripthat =\iscript ^*(u_2)$. If $\iscript$ measures $a\in\escript _1$ and $b\in\escript _2$, we define the \textit{sequential product of $a$ then $b$ relative to} $\iscript$ by $a\sqbrac{\iscript}b=\iscript ^*(b)$. Since $a$ is measured first, the measurement of $a$ can interfere with the measurement of $b$, but not vice-versa. For this reason $a\sqbrac{\iscript}b$ has regular properties as a function of $b$ but not as a function of $a$, in general.

An effect-valued measure is called an \textit{observable} and we say that effects $a,b\in\escript$ \textit{ocexist} if they are simultaneously measured by a single observable. We denote the set of observables on $\escript$ by $\rmob (\escript)$. We define the concept of a bi-observable and say that two observables \textit{coexist} if they are the marginals of a single bi-observable. If $\iscript\in\rmin (\escript _1,\escript _2)$, then $\iscript$ measures a unique observable $\iscripthat$ that satisfies
$s(\iscripthat _x)=\iscript _x(s)(u_2)$ for all $s\in\sscript (\escript _1)$. We interpret $\iscripthat _x$ as the effect that occurs when a measurement of $\iscripthat$ results in the outcome $x$. As with observables, we define bi-instruments and say that two instruments \textit{coexist} if they are the marginals of a single bi-instrument. If
$\iscript ,\jscript\in\rmin (\escript _1,\escript _2)$ coexist, we show that $\iscripthat ,\jscripthat$ coexist. Let $A\in\rmob(\escript _1)$, $B\in\rmob (\escript _2)$ and suppose $\iscript\in\rmin (\escript _1,\escript _2)$ satisfies $\iscripthat =A$. We define the \textit{sequential product of $A$ then $B$ relative to} $\iscript$ to be the
bi-observable on $\escript _1$ given by $\paren{A\sqbrac{\iscript}B}_{xy}=\iscript _x^*(B_y)$. We define $B$ \textit{conditioned by $A$ relative to} $\iscript$ by
$\paren{B\ab{\iscript}A}_y=\sum\limits _x\iscript _x^*(B_y)$. Section~3 also considers repeatable effects and proves that various conditions are equivalent to repeatability.

Section 4 discusses Holevo instruments. Let $\escript _1=\sqbrac{\theta _1,u_1}$, $\escript _2=\sqbrac{\theta _2,u_2}$ and suppose $a\in\escript _1$,
$\beta\in\sscript (\escript _2)$. A \textit{pure Holevo operation} with effect $a$ and state $\beta$ denoted by $\hscript ^{(a,\beta )}\in\oscript (\escript _1,\escript _2)$ has the form $\hscript ^{(a,\beta )}(\alpha )=\alpha (a)\beta$ for all $\alpha\in\sscript (\escript _1)$. We show that $\hscript ^{(\alpha ,\beta )*}(b)=\beta (b)a$ for all $b\in\escript _2$. It follows that $\hscript ^{(a,\beta )*}(u_2)=a$ so $\hscript ^{(a,\beta )}$ measures the effect $a\in\escript _1$. Since the state $\beta\in\sscript (\escript _2)$ can vary, this shows that an effect $a$ can be measured by many operations $\hscript ^{(a,\beta )}$ with $\paren{\hscript ^{(a,\beta )}}^\wedge=a$. If $\beta _i\in\sscript (\escript _2)$ and
$a_i\in\escript _1$, $i=1,2,\ldots ,n$ with $\sum\limits _{i=1}^na_i\le u_1$, then a \textit{mixed Holevo operation} with effects $a_i$ and states $\beta _i$ has the form
\begin{equation*}
\iscript =\sum _{i=1}^n\hscript ^{(a_i,\beta _i)}\in\oscript (\escript _1,\escript _2)
\end{equation*}
We extend this definition to define pure and mixed Holevo instruments in a natural way. The Holevo sequential product of two observables is defined and the marginals of this product are computed. We also define coexistence of observables. We discuss the commutant of two effects and derive its various properties.

Examples are given that illustrate the properties of concepts presented in this article. For instance, Hilbert space effects, observables and instruments are discussed.

\section{Interval Effect Algebras}  
Most statistical theories for physical systems contain two basic primitive concepts, namely effects and states \cite{bgl95,blpy16,dl70,hz12,kra83,nc00}. The effects
correspond to simple \ityes-\itno\ measurements or experiments and the states correspond to preparation procedures that specify the initial conditions of the system being investigated. Usually, each effect $a$ and state $s$ experimentally determine a probability $F(a,s)$ that the effect $a$ occurs when the system has been prepared in the state $s$. For a given physical system, denote its set of possible effects by $\escript$ and its set of possible states by $\sscript$. In a reasonable statistical theory, the probability function $F$ satisfies two axioms given in the following definition \cite{bgp00,gp98,gud99,gpbb00}.

An \textit{effect-state space} is a triple $(\escript ,\sscript ,F)$ where $\escript$ and $\sscript$ are nonempty sets and $F$ is a mapping from $\escript\times\sscript$ into
$\sqbrac{0,1}\subseteq\real$ satisfying:
\begin{list}
{(E\arabic{cond})}{\usecounter{cond}
\setlength{\rightmargin}{\leftmargin}}
\item There exist elements $\theta ,u\in\escript$ such that $F(\theta ,s)=0$, $F(u,s)=1$ for every $s\in\sscript$.
\item If $F(a,s)\le F(b,s)$ for every $s\in\sscript$, then there exists a unique $c\in\escript$ such that $F(a,s)+F(c,s)=F(b,s)$ for every $s\in\sscript$.
\end{list}

The elements $\theta ,u$ in (ES1) correspond to the null effect that never occurs and the unit effect that always occurs, respectively. Condition~(ES2) postulates that if $a$ has a smaller probability of occurring than $b$ in every state, then there exists a unique effect $c$ which when combined with $a$ gives the probability that $b$ occurs in every state. In the sequel, we use the notation $s(a)=F(a,s)$. The following result was proved in \cite{gud99}.

\begin{thm}    
\label{thm21}
Let $(\escript ,\sscript ,F)$ be an effect-state space. If $s(a)+s(b)\le 1$ for every $s\in\sscript$, then there exists a unique $c\in\escript$ such that $s(c)=s(a)+s(b)$ for every $s\in\sscript$.
\end{thm}

If $s(a)+s(b)\le 1$ for every $s\in\sscript$, we write $a\perp b$ and the unique $c\in\escript$ in Theorem~\ref{thm21} is denoted $c=a+b$. We also write $a\le c$ and it is shown in \cite{bgp00,gud99,gpbb00} that $a\le c$ if and only if $s(a)\le s(c)$ for all $s\in\sscript$. When this condition holds, we say that $\sscript$ is
\textit{order-determining} on $\escript$. Moreover, it is shown in \cite{bgp00,gud99,gpbb00} that $(\escript ,\sscript ,F)$ is an effect-state space if and only if
$(\escript ,\theta ,u,+)$ is an effect algebra with order-determining set of states $\sscript$. We emphasize that a \textit{state} on $(\escript ,\theta ,u,+)$ is a function $s\colon\escript\to\sqbrac{0,1}$ satisfying $s(u)=1$ and $s(a+b)=s(a)+s(b)$ when $a\perp b$.

We now show that a natural convex structure can be defined on an effect-state space $(\escript ,\sscript ,F)$. We say that $(\escript ,\sscript ,F)$ is a
\textit{convex effect-state space} if for every $a\in\escript$, $\lambda\in\sqbrac{0,1}\subseteq\real$ there exists an element $\lambda a\in\escript$ such that
$s(\lambda a)=\lambda s(a)$ for every $s\in\sscript$. It follows that $\lambda a$ is unique and we interpret $\lambda a$ as the effect $a$ attenuated by the factor
$\lambda$. It is shown in \cite{gp98,gud99} that $(\escript ,\sscript ,F)$ is a convex effect-state space if and only if $(\escript ,\theta ,u,+)$ is a convex effect algebra with order-determining set of states $\sscript$.

We now introduce an equivalent structure that we find convenient to work with. This structure enables us to extend convexity to linearity. Let $V$ be a real vector space with zero $\theta$. A subset $K$ of $V$ is a \textit{positive cone} if $\real ^+K\le K$, $K+K\le K$ and $K\cap (-K)=\brac{\theta}$ For $x,y\in V$, we define $x\le y$ if $y-x\in K$. Then $\le$ is a partial order on $V$ and we call $(V,K)$ an \textit{ordered linear space} with positive cone $K$. We say that $K$ is \textit{generating} if $V=K-K$. Let
$u\in K$ with $u\ne\theta$ and form the interval
\begin{equation*}
\sqbrac{\theta ,u}=\brac{x\in K\colon x\le u}
\end{equation*}
For $x,y\in\sqbrac{\theta ,u}$ we write $x\perp y$ if $x+y\le u$ and in this case we have $x+y\in\sqbrac{\theta ,u}$. It is easy to check that $\sqbrac{\theta ,u}$ is a convex subset of $K$ in the sense that if $x,y\in\sqbrac{\theta ,u}$ and $\lambda\in\sqbrac{0,1}$, then $\lambda x+(1-\lambda )y\in\sqbrac{\theta ,u}$. It follows that 
$\paren{\sqbrac{\theta ,u},\theta ,u,+}$ is a convex effect algebra that we call an \textit{interval effect algebra} \cite{gp98,gud99}. We say that $\sqbrac{\theta ,u}$
\textit{generates} $K$ if $K=\real ^+\sqbrac{\theta ,u}$ and say that $\sqbrac{\theta ,u}$ \textit{generates} $V$ if $\sqbrac{\theta ,u}$ generates $K$ and $K$ generates $V$. The following result was proved in \cite{gp98,gud99}.

\begin{thm}    
\label{thm22}
$(\escript ,\sscript ,F)$ is a convex effect-state space if and only if $\paren{\sqbrac{\theta ,u},\theta ,u,+}$ is an interval effect algebra with an order-determining set of states $\sscript$ and $\sqbrac{\theta ,u}$ generates $V$.
\end{thm}

Because of this theorem, we can employ interval effect algebras as our basic mathematical structure. This will enable us to simplify our notation and use the power of linear spaces.

\section{Instruments on Interval Effect Algebras}  
In the sequel, we assume that $\escript =\sqbrac{\theta ,u}$ where $\paren{\sqbrac{\theta ,u},\theta ,u,+}$ is an interval effect algebra on $(V,K)$ with order-determining set of states $\sscript (\escript )$ and $\escript$ generates $V$. Without loss of generality, we can assume that $\sscript (\escript )$ is a convex set in the sense that
$\lambda\in\sqbrac{0,1}$, $s_1,s_2\in\sscript (\escript )$ imply $\lambda s_1+(1-\lambda )s_2\in\sscript (\escript )$. It is shown in \cite{bgp00,gpbb00}, that any
$s\in\sscript (\escript )$ extends to a unique linear functional from $V$ to $\real$. If $S$ and $T$ are convex sets, a function $f\colon S\to T$ is \textit{affine} if for
$\lambda _i\in\sqbrac{0,1}$ with $\sum\limits _{i=1}^n\lambda _i=1$, we have $f\paren{\sum\limits _{i-1}^n\lambda _ix_i}=\sum\lambda _if(x_i)$. We call $\escript$
\textit{unrestricted} if for every affine $f\colon\sscript (\escript )\to\sqbrac{0,1}$ there exists an effect $a\in\escript$ such that $f(s)=s(a)$ for all $s\in\sscript (\escript )$. Just as order-determining specifies there are enough states to determine the order of effects, unrestricted specifies there are enough effects to determine the affine functions on states. We shall assume in this section that $\escript$ is unrestricted.

A function $f\colon\escript\to\sqbrac{0,1}$ is a \textit{substate} if $f(a)=\lambda s(a)$ for all $a\in\escript$ where $\lambda\in\sqbrac{0,1}$ and $s\in\sscript (\escript )$. We denote the set of substates by $\rmsub\paren{\sscript (\escript )}$. Let $\escript _1=\sqbrac{\theta _1,u_1}$, $\escript _2=\sqbrac{\theta _2,u_2}$ be interval effect algebras satisfying the requirements of the previous paragraph. An affine map $\iscript\colon\sscript (\escript _1)\to\rmsub\paren{\sscript (\escript )}$ is called an
\textit{operation from $\escript _1$ to} $\escript _2$ and an operation $\iscript\colon\sscript (\escript _1)\to\sscript (\escript _2)$ is called a \textit{channel}
\cite{bgl95,blpy16,dl70,kra83}. Thus, an operation $\iscript$ is a channel if and only if $\iscript (s)(u_2)=1$ for all $s\in\sscript (\escript _1)$. This definition is more general than for the usual Hilbert space effect algebra $\escript (H)$ because in $\escript (H)$, $\iscript$ is required to be completely positive \cite{hz12,nc00} and this concept is not applicable for a general $\escript$. We denote the set of operations from $\escript _1$ to $\escript _2$ by $\oscript (\escript _1,\escript _2)$ and we use the notation
$\oscript (\escript )=\oscript (\escript ,\escript )$. The map $\iscript ^*\colon\escript _2\to\escript _1$ in the next lemma is called the \textit{dual} of
$\iscript\in\oscript (\escript _1,\escript _2)$ \cite{gud120,gud220,gud22,gud24}.

\begin{lem}    
\label{lem31}
{\rm{(i)}}\enspace If $\iscript\in\oscript (\escript _1,\escript _2)$, there exists a unique affine map $\iscript ^*\colon\escript _2\to\escript _1$ that satisfies
$s\sqbrac{\iscript ^*(a)}=\iscript (s)(a)$ for all $s\in\sscript (\escript _1)$, $a\in\escript _2$.
{\rm{(ii)}}\enspace $\iscript$ is a channel if and only if $\iscript ^*(u_2)=u_1$.
\end{lem}
\begin{proof}
(i)\enspace For $a\in\escript _2$ define the affine map $f_a\colon\sscript (\escript _1)\to\sqbrac{0,1}$ by $f_a(s)=\iscript (s)(a)$. Since $\escript _1$ is unrestricted, there exists $b\in\escript _1$ such that
\begin{equation*}
\iscript (s)(a)=f_a(s)=s(b)
\end{equation*}
for all $s\in\sscript (\escript _1)$. Now $b$ is unique because if $s(b_1)=s(b)$ for all $s\in\sscript (\escript _1)$, then since $\sscript (\escript _1)$ is order-determining,
$b_1=b$. Define the map $\iscript ^*\colon\escript _2\to\escript _1$ by $\iscript ^*(a)=b$. Then $\iscript ^*$ is affine because if $\lambda _i\in\sqbrac{0,1}$,
$\sum\limits _{i=1}^n\lambda _i=1$, we obtain for $a_i\in\escript _2$, $i=1,2,\ldots ,n$ that 
\begin{align*}
s\sqbrac{\iscript ^*\paren{\sum _{i=1}^n\lambda _ia_i}}&=\iscript (s)\paren{\sum _{i=1}^n\lambda _ia_i}=\sum _{i=1}^n\lambda _i\iscript (s)(a)
   =\sum _{i=1}^n\lambda _is\sqbrac{\iscript ^*(a_i)}\\
   &=s\sqbrac{\sum _{i=1}^n\lambda _i\iscript ^*(a_i)}
\end{align*}
for all $s\in\sscript (\escript _1)$. Since $\sscript (\escript _1)$ is order-determining,
$\iscript ^*\paren{\sum\limits _{i=1}^n\lambda _ia_i}=\sum\limits _{i=1}^n\lambda _i\iscript ^*(a_i)$. Finally, $\iscript ^*$ is unique because
$s\sqbrac{\iscript ^*(a)}=\iscript (s)(a)$ for all $s\in\sscript (\escript _1)$ and $\sscript (\escript _	1)$ is order-determining.
(ii)\enspace If $\iscript$ is a channel, then
\begin{equation*}
s\sqbrac{\iscript ^*(u_2)}=\iscript (s)(u_1)=1
\end{equation*}
for all $s\in\sscript (\escript _1)$. Therefore, $\iscript ^*(u_2)=u_1$. Conversely, if $\iscript ^*(u_2)=u_1$, then 
\begin{equation*}
\iscript (s)(u_2)=s\paren{\iscript ^*(u_2)}=1
\end{equation*}
for all $s\in]\sscript (\escript _1)$ so $\iscript (s)\in\sscript (\escript _2)$. Hence, $\iscript$ is a channel.
\end{proof}

We say that $\iscript\in\oscript (\escript _1,\escript _2)$ \textit{measures} the effect $\iscripthat =\iscript ^*(u_2)$. Notice that $\iscripthat\in\escript _1$ is the unique effect satisfying $s(\iscripthat\,)=\iscript (s)(u_2)$ for all $s\in\sscript (\escript _1)$ where $s(\iscripthat\,)$ is the probability that $\iscripthat$ occurs in the state $s$. Although an operation measures a unique effect, we shall see that an effect is measured by many operations. If $\iscript$ is employed to make a measurement and $\iscripthat$ occurs in a state $s$ where $s(\iscripthat\,)\ne 0$, then the \textit{updated} state is defined by $s'=\iscript (s)/s(\iscripthat\,)\in\sscript (\escript _2)$. We see that $s'$ is a state because
\begin{equation*}
s'(u_2)=\tfrac{1}{s(\iscripthat\,)}\iscript (s)(u_2)=1
\end{equation*}
If $\iscript$ measures $a\in\escript _1$ and $b\in\escript _2$, we define the \textit{sequential product of $a$ then $b$ relative to} $\iscript$ to be
$a\sqbrac{\iscript}b=\iscript ^*(b)$ \cite{gg02,gn01,gud22}. Since $a$ is measured first, the $a$ measurement can interfere with the $b$ measurement but not vice versa. For this reason $a\sqbrac{\iscript}b$ has regular properties as a function of $b$ but not as a function of $a$, in general. The physical meaning of $a\sqbrac{\iscript}b$ is the following. The system $S$ is initially in a state $s\in\sscript (\escript _1)$ and the effect $a\in\escript _1$ is measured using the operation
$\iscript\in\oscript (\escript _1,\escript _2)$. The state is updated to the state $s'\in\sscript (\escript _2)$ and $b\in\escript _2$ is measured which results in the effect
$\iscript _1^*(b)\in\escript _1$. The probability that $\iscript ^*(b)$ occurs is 
\begin{equation*}
s\sqbrac{\iscript ^*(b)}=\iscript (s)(b)=s(\iscripthat\,)\,\tfrac{\iscript (s)}{s(\iscripthat\,)}(b)=s(\iscripthat\,)s'(a)s'(b)
\end{equation*}
Thus, the probability that $a\sqbrac{\iscript}b=\iscript ^*(b)$ occurs in the state $s$ is the probability that $a$ occurs in the state $s$ times the probability that $b$ occurs in the updated state $s'$.

\begin{example}  
The simplest example of an operation $\iscript\in\oscript (\escript _1,\escript _2)$ is a \textit{constant channel} $\iscript (s)=s_2\in\sscript (\escript _2)$ for all
$s\in\sscript (\escript _1)$. Then $\iscript ^*\colon\escript _2\to\escript _1$ satisfies
\begin{equation*}
s\sqbrac{\iscript ^*(a)}=\iscript (s)(a)=s_2(a)=s\sqbrac{s_2(a)u_1}
\end{equation*}
for every $s\in\sscript (\escript _1)$. Hence, $\iscript ^*(a)=s_2(a)u_1$ for all $a\in\escript _2$. The effect measured by $\iscript$ is $\iscripthat =\iscript ^*(u_2)=u_1$. Since there are many constant channels corresponding to the various $s_2\in\sscript (\escript _2)$, $u_1$ is measured by many operations. Although
$u_1\sqbrac{\iscript}u_2=\iscript ^*(u_2)=u_1$ we have
\begin{equation*}
u_1\sqbrac{\iscript }b=\iscript ^*(b)=s_2(b)u_1
\end{equation*}
and this is not $b$, in general, as one might expect.\hfill\qedsymbol
\end{example}

A (finite) \textit{observable} on $\escript =\sqbrac{\theta ,u}$ is a finite set $A=\brac{A_x\colon x\in\Omega _A}\subseteq\escript$ such that
$\sum\limits _{x\in\Omega _A}A_x=u$. We denote the set of observables on $\escript$ by $\rmob (\escript )$. If $A\in\rmob (\escript )$, we call $\Omega _A$ the
\textit{outcome space} for $A$ and we call $A(\Delta )=\sum\limits _{x\in\Delta}A_x$, where $\Delta\subseteq\Omega _A$, an \textit{effect-valued measure}. For
$s\in\sscript (\escript )$, $A\in\rmob (\escript )$, we call $s(A_x)$ the $s$-\textit{probability that $A$ has outcome} $x$ and $\Phi _s^A(s)=s\sqbrac{A(\Delta )}$,
$\Delta\subseteq\Omega _A$ is the $s$-\textit{distribution} of $A$. Two effects $a,b\in\escript$ \textit{coexist} if there exists an observable $A\in\rmob (\escript )$,
$\Delta _1,\Delta _2\subseteq\Omega _A$ such that $a=A(\Delta _1)$, $b=A(\Delta _2)$. In this way $a,b$ are simultaneously measured by $A$. A bi-observable is an observable $C$ that satisfies $\Omega _C=\Omega _1\times\Omega _2$ and we write
\begin{equation*}
C=\brac{C_{xy}\colon x\in\Omega _1, y\in\Omega _2}=\brac{C_{xy}\colon (x,y)\in\Omega _1\times\Omega _2}
\end{equation*}
Two observables $A=\brac{A_x\colon x\in\Omega _A}$, $B=\brac{B_y\colon y\in\Omega _B}$ in $\rmob (\escript )$ \textit{coexist} if there exists a
\textit{joint bi-observable} \cite{gud120,gud220,gud24}
\begin{equation*}
C=\brac{C_{xy}\colon (x,y)\in\Omega _A\times\Omega _B}\in\rmob (\escript )
\end{equation*}
such that $\sum\limits _{y\in\Omega _B}C_{xy}=A_x$ for all $x\in\Omega _A$ and $\sum\limits _{x\in\Omega _A}C_{xy}=B_y$ for all $y\in\Omega _B$. If $A,B$ coexist, then $A_x,B_y\in\escript$ coexist for all $x\in\Omega _A$, $y\in\Omega _B$ because $A_x=C\paren{\brac{x}\times\Omega _B}$ and
$B_y=\paren{\Omega _A\times\brac{y}}$. If $C$ is a bi-observable with $\Omega _C=\Omega _1\times\Omega _2$ we call the observables $C^1,C^2$ given by
$C_x^1=\sum\limits _{y\in\Omega _2}C_{xy}$, $C_y^2=\sum\limits _{x\in\Omega _1}C_{xy}$, the \textit{marginals} of $C$. It follows that $C^1,C^2$ coexist.

A (finite) \textit{instrument} from $\escript _1$ to $\escript _2$ is a finite set $\iscript =\brac{\iscript _x\colon x\in\Omega _\iscript}\subseteq\oscript (\escript _1,\escript _2)$ such that $\iscriptbar =\sum\limits _{x\in\Omega _\iscript}\iscript _x$ is a channel. Since $\iscriptbar$ is a channel, we have for all $s\in\sscript (\escript _1)$ that
\begin{equation*}
s\sqbrac{\,\iscriptbar\,^*(u_2)}=\iscriptbar (s)(u_2)=1
\end{equation*}
so $\iscriptbar\,^*(u_2)=u_1$. We call $\Omega _\iscript$ the \textit{outcome space} of $\iscript$ and think of $\iscript _x$ as the operation that occurs when a measurement of $\iscript$ results in outcome $x$. We denote the set of instruments from $\escript _1$ to $\escript _2$ by $\rmin (\escript _1,\escript _2)$. We call
$\iscript (\Delta )=\sum\limits _{x\in\Delta}\iscript _x$, $\Delta\subseteq\Omega _\iscript$, an \textit{operation-valued measure}. If $s\in\sscript (\escript _1)$, then
$\iscript _x(s)(u_2)$ is the probability that $\iscript _x$ occurs in the state $s$ and
\begin{equation*}
\Phi _s^\iscript (\Delta )=\sum _{x\in\Delta}\iscript _x(s)(u_2)
\end{equation*}
is the $s$-\textit{distribution} of $\iscript$. Notice that $\Phi _s^\iscript$ is a probability distribution because
\begin{equation*}
\Phi _s^\iscript (\Omega _\iscript )=\sum _{x\in\Omega _\iscript}\iscript _x(s)(u_2)=\iscripthat (s)(u_2)=1
\end{equation*}
Letting $\iscripthat _x=\iscript _x^*(u_2)\in\escript _1$ we have that
\begin{equation*}
s(\iscripthat _x)=s\sqbrac{\iscript _x^*(u_2)}=\iscript _x(s)(u_2)
\end{equation*}
for all $s\in\sscript (\escript _1)$. Since $\iscriptbar$ is a channel, we obtain
\begin{equation*}
s\paren{\sum _{x\in\Omega _\iscript}\iscripthat _x}=\sum _{x\in\Omega _\iscript}\iscript _x(s)(u_2)=\iscriptbar (s)(u_2)=1
\end{equation*}
for all $s\in\sscript (\escript _1)$. It follows that $\sum\limits _{x\in\Omega _\iscript}\iscripthat _x=u_1$ so
\begin{equation*}
\iscripthat =\brac{\iscripthat _x\colon x\in\Omega _{\iscripthat}=\Omega _\iscript}
\end{equation*}
is an observable on $\escript _1$. Then $\iscripthat$ is the unique observable on $\escript _1$ that satisfies $s(\iscripthat _x)=\iscript _x(s)(u_2)$ for all
$s\in\sscript (\escript _1)$ and we call $\iscripthat$ the observable \textit{measured} by $\iscript$. As with effects, although an instrument measures a unique observable, as we shall see, any observable is measured by many instruments and the set of instruments from $\escript _1$ to $\escript _2$ is denoted $\rmin (\escript _1,\escript _2)$
\cite{gud220,gud22,gud23,gud24}. If $\iscript\in\rmin (\escript _1,\escript _2)$, $\jscript\in\rmin (\escript _2,\escript _3)$ we define their
\textit{sequential product} $\iscript\circ\jscript\in\rmin (\escript _1,\escript _3)$ by $(\iscript\circ\jscript )_{xy}(s)=\jscript _y\paren{\iscript _s(x)}$ for all
$s\in\sscript (\escript _1)$. Then
\begin{equation*}
(\iscript\circ\jscript )_{xy}^*(b)=\iscript _x^*\paren{\jscript _y^*(b)}
\end{equation*}
for all $b\in\escript _3$ so $(\iscript\circ\jscript )^*=(\jscript\circ\iscript )^*$. We also have
\begin{equation*}
(\iscript\circ\jscript )_{xy}^\wedge =(\iscript _x\circ\jscript _y)^*(u_3)=\iscript _x^*\paren{\jscript _y^*(u_3)}=\iscript _x^*(\jscripthat _y)
\end{equation*}

\begin{example}  
The most important example of an interval effect algebra is a Hilbert space effect algebra $\escript (H)$ which is the basic mathematical structure for traditional quantum theory \cite{hz12,nc00}. For simplicity, we assume that the Hilbert space $H$ is finite dimensional. In this case, the real linear space $V$ is the set of self-adjoint operators
$\lscript _S(H)$ on $H$ and the positive cone $K$ is the set of positive operators 
\begin{equation*}
K=\brac{T\in V\colon T\ge 0}\subseteq V
\end{equation*}
where $0$ is the zero operator. Letting $I$ be the identity operator we have the set of effects $\escript (H)=\brac{T\in V\colon 0\le T\le I}$, A \textit{density operator} is an operator $\rho\in K$ with trace $\trace (\rho )=1$. It is well known that if $\dim H\ge 3$ then $s\in\sscript\paren{\escript (H)}$ if and only if there exists a density operator
$\rho$ with $s$ given by the Born rule $s(a)=\trace (\rho a)$ for all $a\in\escript (H)$ \cite{hz12,nc00}. It is also known that $\sscript (\escript )$ is order-determining and unrestricted \cite{hz12}. An observable
\begin{equation*}
A=\brac{A_x\colon x\in\Omega _A}\subseteq\escript (H)
\end{equation*}
with $\sum\limits _xA_x=I$ is called a \textit{normalized, positive-operator valued measure} (POVM). A \textit{Hilbert space operation} on $\escript (H)$ is an affine,
trace-nonincreasing, completely positive map $\iscript\colon\sscript\paren{\escript (H)}\to\rmsub\paren{\sscript\paren{\escript (H)}}$ \cite{hz12}. A trace-preserving operation is a channel. Every operation on $\escript (H)$ has the form $\iscript (\rho )=\sum K_i\rho K_i^*$, where $K_i$ are linear operators on $H$ with $\sum K_i^*K_i\le I$
\cite{hz12} and $\iscript$ is a channel if and only if $\sum K_i^*K_i=I$. An important example is a \textit{L\"uders channel} which has the form
$\iscript (\rho )=\sum A_i^{1/2}\rho A_i^{1/2}$ where $A=\brac{A_i \colon i=1,2,\ldots ,n}$ is an observable. An \textit{instrument} on $\escript (H)$ is a set
\begin{equation*}
\iscript =\brac{\iscript _x\colon x\in\Omega _\iscript}\subseteq\oscript\paren{\escript (H)}
\end{equation*}
such that $\sum\iscript _x$ is a channel \cite{dl70,kra83}. The Holevo operations and instruments discussed in Section~4 apply to this example as well
\cite{hol82,hol98}.\hfill\qedsymbol
\end{example}

As with observables, a \textit{bi-instrument} is an instrument $\kscript$ that satisfies $\Omega _\kscript =\Omega _1\times\Omega _2$ and we write
$\kscript =\brac{\kscript _{xy}\colon (x,y)\in\Omega _1\times\Omega _2}$. Two instruments $\iscript ,\jscript\in\rmin (\escript _1,\escript _2)$ \textit{coexist} if there exists a
\textit{joint bi-instrument} $\kscript\in\rmin (\escript _1,\escript _2)$ such that $\sum\limits _{y\in\Omega _2}\kscript _{xy}=\iscript _x$ and
$\sum\limits _{x\in\Omega _1}\kscript _{xy}=\jscript _y$ \cite{gud220,gud22,gud23,gud24}.

\begin{lem}    
\label{lem32}
If $\iscript ,\jscript\in\rmin (\escript _1,\escript _2)$ coexist, then $\iscripthat$, $\jscripthat$ coexist.
\end{lem}
\begin{proof}
Let $\kscript\in\rmin (\escript _1,\escript _2)$ be a joint instrument for $\iscript$, $\jscript$. Then
\begin{equation*}
\kscripthat =\brac{\kscripthat _{xy}\colon (x,y)\in\Omega _1\times\Omega _2}
\end{equation*}
is a bi-observable in $\rmob (\escript _1)$ and for all $s\in\sscript (\escript _1)$, $x\in\Omega _1$ we have
\begin{align*}
s\paren{\sum _{y\in\Omega _2}\kscripthat _{xy}}&=\sum _{y\in\Omega _2}\sqbrac{s(\kscripthat _{xy})}=\sum _{y\in\Omega _2}\sqbrac{\kscript _{xy}(s)(u_2)}
   =\sqbrac{\sum _{y\in\Omega _2}\kscript _{xy}(s)}(u_2)\\
   &=\iscript _x(s)(u_2)=s(\iscripthat _x)
\end{align*}
and similarly, $s\paren{\sum\limits _{x\in\Omega _1}\kscripthat _{xy}}=s(\jscripthat _y)$. Therefore, $\sum\limits _{y\in\Omega _2}\kscripthat _{xy}=\iscripthat _x$ and\newline
$\sum\limits _{x\in\Omega _1}\kscripthat _{xy}=\jscripthat _y$ so $\kscripthat$ is a joint observable for $\iscripthat ,\jscripthat$. Hence $\iscripthat ,\jscripthat$ coexist
\end{proof}

As we shall see, simple examples show that the converse of Lemma~\ref{lem32} does not hold.

Let $A\in\rmob (\escript _1)$, $B\in\rmob (\escript _2)$ and suppose $\iscript\in\rmin (\escript _1,\escript _2)$ satisfies $\iscripthat =A$. The 
\textit{sequential product of $A$ then $B$ relative to} $\iscript$ is the bi-observable on $\escript _1$ given by $\paren{A\sqbrac{\iscript}B}_{xy}=\iscript _x^*(B_y)$. Notice that $A\sqbrac{\iscript}B$ is indeed a bi-observable because for $s\in\sscript (\escript _1)$ we have
\begin{align*}
s\sqbrac{\sum _{x,y}\paren{A\sqbrac{\iscript}B}_{xy}}&=s\sqbrac{\sum _{x,y}\iscript _x^*(B_y)}=s\sqbrac{\sum _x\iscript _x^*(u_2)}=s\sqbrac{\,\iscriptbar\,^*(u_2)}\\
   &=\iscriptbar (s)(u_2)=1
\end{align*}
so $\sum\limits _{x,y}\paren{A\sqbrac{\iscript}B}_{xy}=u_1$. We define $B$ \textit{conditioned by $A$ relative to} $\iscript$ as
\begin{equation*}
\paren{B\ab{\iscript}A}_y=\sum _x\paren{A\sqbrac{\iscript}B}_{xy}=\sum _x\iscript _x^*(B_y)=\iscriptbar\,^*(B_y)
\end{equation*}
for all $y\in\Omega _B$, Again, $\paren{B\ab{\iscript}A}$ is an observable because
\begin{equation*}
\sum _y\paren{B\ab{\iscript}A}_y=\iscriptbar\,^*\paren{\sum _yB_y}=\iscriptbar\,^*(u_2)=u_1
\end{equation*}
Since $\sum\limits _x\paren{A\sqbrac{\iscript}B}_{xy}=\paren{B\ab{\iscript}A}_y$ and
\begin{equation*}
\sum _y\paren{A\sqbrac{\iscript}B}_{xy}=\sum _y\iscript _x^*(B_y)=\iscript _x^*\paren{\sum _yB_y}=\iscript _x^*(u_2)=\iscripthat _x=A_x
\end{equation*}
we have that $A$ and $\paren{B\ab{\iscript}A}$ coexist with joint observable $A\sqbrac{\iscript}B$ \cite{gud220,gud22,gud23,gud24}.

If $a\in\escript =\sqbrac{\theta ,u}$, $\iscript\in\oscript (\escript )$ with $\iscripthat =a$, then $a$ is $\iscript$-\textit{repeatable} \cite{hz12} if $\iscript ^*(a)=a$. Thus, $a$ is
$\iscript$-repeatable if and only if
\begin{equation*}
\iscript ^*(u)=\iscript ^*(a)=a
\end{equation*}
The \textit{complement} of $a\in\escript$ is the effect $a'=u-a$. Then $a'$ is the unique effect satisfying $s(a')=1-s(a)$ for all $s\in\sscript (\escript )$.

\begin{thm}    
\label{thm33}
The following statements are equivalent
{\rm{(i)}}\enspace $a$ is $\iscript$-repeatable.
{\rm{(ii)}}\enspace $a\sqbrac{\iscript}a=a$,
{\rm{(iii)}}\enspace $a\sqbrac{\iscript}a'=\theta$,
{\rm{(iv)}}\enspace $a\sqbrac{\iscript}b=\theta$ whenever $a\perp b$,
{\rm{(v)}}\enspace $\iscript^*(b)\le\iscript ^*(a)$ for all $b\in\escript$,
{\rm{(vi)}}\enspace $\iscript\paren{\iscript (s)}(u)=\iscript (s)(u)$ for all $s\in\sscript (\escript )$ and $\iscript ^*(u)=a$.
\end{thm}
\begin{proof}
(i)$\Leftrightarrow$(ii)\enspace If $a$ is $\iscript$-repeatable, then $\iscript ^*(a)=\iscript ^*(u)=a$ so
\begin{equation*}
a\sqbrac{\iscript}a=\iscript ^*(a)=\iscript ^*(u)=a
\end{equation*}
Conversely, if $a\sqbrac{\iscript}a=a$, then $a=\iscript ^*(a)=\iscript ^*(u)$ so $a$ is $\iscript$-repeatable.\newline
(i)$\Leftrightarrow$(iii)\enspace $a\sqbrac{\iscript}a'=\theta$ if and only if $\iscript ^*(a')=\theta$ and $\iscript ^*(u)=a$. But these are equivalent to
\begin{equation*}
\theta =\iscript ^*(u-a)=\iscript ^*(u)-\iscript ^*(a)=\iscript ^*(u)-a
\end{equation*}
which is equivalent to $a$ being $\iscript$-repeatable.
(iii)$\Leftrightarrow$(iv)\enspace Suppose $a\sqbrac{\iscript}a'=\theta$ and $a\perp b$. Then $b\le a'$ so that
\begin{equation*}
a\sqbrac{\iscript}b=\iscript ^*(b)\le\iscript ^*(a')=a\sqbrac{\iscript}a'=\theta
\end{equation*}
Hence, (iv) holds. Conversely, if $a\sqbrac{\iscript}b=\theta$ whenever $a\perp b$, then since $a\perp a'$ we have $a\sqbrac{\iscript}a'=\theta$.
(i)$\Leftrightarrow$(v)\enspace If $\iscript ^*(b)\le\iscript ^*(a)$ for all $b\in\escript$, then $a=\iscript ^*(u)\le\iscript ^*(a)$. Since $\iscript ^*(a)\le\iscript ^*(u)=a$ we have
$\iscript ^*(a)=a$. Hence, $a$ is $\iscript$-repeatable. Conversely, if $a$ is $\iscript$-repeatable, then 
\begin{equation*}
\iscript ^*(b)\le\iscript ^*(u)=a=\iscript ^*(a)
\end{equation*}
for all $b\in\escript$.
(i)$\Leftrightarrow$(vi)\enspace If (vi) holds, then for every $s\in\sscript (\escript )$ we obtain
\begin{equation*}
s(a)=s\sqbrac{\iscript ^*(u)}=\iscript (s)(u)=\iscript (s)\paren{\iscript ^*(u)}=\iscript (s)(a)=s\sqbrac {\iscript ^*(a)}
\end{equation*}
Hence, $\iscript ^*(a)=\iscript ^*(u)=a$ so $a$ is $\iscript$-repeatable. Conversely, if $a$ is $\iscript$-repeatable, then for all $s\in\sscript (\escript )$ we have
$\iscript ^*(u)=a$ and 
\begin{align*}
\iscript (s)(u)&=s\sqbrac{\iscript ^*(u)}=s(a)=s\sqbrac{\iscript ^*(a)}=\iscript (s)(a)\\
   &=\iscript (a)\paren{\iscript ^*(u)}=\iscript\paren{\iscript (s)}(u)
\end{align*}
so (iv) holds.
\end{proof}

An effect $a\in\escript$ is \textit{repeatable} if there exists an $\iscript\in\oscript (\escript )$ such that $a$ is $\iscript$-repeatable \cite{hz12}

\begin{cor}    
\label{cor34}
An effect $a\in\escript$ is repeatable if and only if $a=\theta$ or there exists an $s\in\sscript (\escript )$ such that $s(a)=1$.
\end{cor}
\begin{proof}
Notice that $\theta$ is repeatable. Suppose $a$ is $\iscript$-repeatable and $a\ne\theta$. Then $\iscript ^*(u)=\iscript ^*(a)=a$. Let $s_0\in\sscript (\escript )$ satisfying
$\iscript (s_0)\ne 0$ and let $s_1=\iscript (s_0)/\iscript (s_0)(u)$. Then $s_1\in\sscript (\escript )$ and by Theorem~\ref{thm33} (vi) we have
\begin{equation*}
1=\frac{\iscript (s_0)(u)}{\iscript (s_0)(u)}=\frac{\iscript\sqbrac{\paren{\iscript (s_0)}}(u)}{\iscript (s_0)(u)}=\iscript (s_1)(u)=s_1\sqbrac{\iscript ^*(u)}=s_1(a)
\end{equation*}
Conversely, suppose $s_1(a)=1$. Define $\iscript\in\oscript (\escript )$ by $\iscript (s)=s(a)s_1$. We then obtain for all $b\in\escript$, $s\in\sscript (\escript )$ that
\begin{equation*}
s\sqbrac{\iscript ^*(b)}=\sqbrac{\iscript (s)}(b)=s(a)s_1(b)=s\sqbrac{s_1(b)a}
\end{equation*}
Hence, $\iscript ^*(b)=s_1(b)a$ for all $b\in\escript$. We then obtain
\begin{equation*}
\iscript ^*(u)=a=s_1(a)a=\iscript ^*(a)
\end{equation*}
so $a$ is $\iscript$-repeatable
\end{proof}

\section{Holevo Instruments}  
In this section all effect algebras will be assumed to be interval effect algebras with an order-determining set of states. Let $\escript _1=\sqbrac{\theta _1,u_1}$,
$\escript _2=\sqbrac{\theta _2,u_2}$ and suppose $a\in\escript _1$, $\beta\in\sscript (\escript _2)$. A \textit{pure Holevo operation with effect $a$ and state} $\beta$ denoted by $\hscript ^{(a,\beta )}\in\oscript (\escript _1,\escript _2)$ has the form $\hscript ^{(a,\beta )}(\alpha )=\alpha (a)\beta$ for all $\alpha\in\sscript (\escript _1)$. Since
\begin{equation*}
\alpha\sqbrac{\hscript ^{(a,\beta )*}(b)}=\sqbrac{\hscript ^{(a,\beta )}(\alpha )}(b)=\alpha (a)\beta (b)=\alpha\sqbrac{\beta (b)a}
\end{equation*}
for all $\alpha\in\sscript (\escript _1)$ we have $\hscript ^{(a,\beta )*}(b)=\beta (b)a$ for all $b\in\escript _2$. Then $\hscript ^{(a,\beta )*}(u_2)=a$ so
$\hscript ^{(a,\beta )}$ measures the effect $a\in\escript _1$. Since the state $\beta\in\sscript (\escript _2)$ can vary, this shows that an effect can be measured by many operations $\hscript ^{(a,\beta )}$ with $(\hscript ^{(a,\beta )})^\wedge = a$. We also see that $\hscript ^{(a,\beta )}$ is a channel if and only if $a=u_1$ and in this case,
$\hscript ^{(u_1,\beta )}(\alpha )=\beta$ for all $\alpha\in\sscript (\escript _1)$ and $\hscript ^{(u_1,\beta )*}(b)=\beta (b)u_1$ for all $b\in\escript _2$.

If $\beta _i\in\sscript (\escript _2)$ and $a_i\in\escript _1$, $i=1,2,\ldots ,n$ with $\sum\limits _{i=1}^na_i\le u_1$, then a
\textit{mixed Holevo operation with effects $a_i$ and states} $\beta _i$ has the form \cite{hol82,hol98}
\begin{equation*}
\iscript =\sum _{i=1}^n\hscript ^{(a_i,\beta _i)}\in\oscript (\escript _1,\escript _2)
\end{equation*}
We see that $\iscript$ is indeed an operation because $\iscript$ is affine and for all $\alpha\in\sscript (\escript _1)$ we have
\begin{align*}
\iscript (\alpha )(u_2)&=\sum _{i=1}^n\hscript ^{(a_i,\beta _i)}(\alpha )(u_2)=\sum _{i=1}^n\alpha (a_i)\beta _i(u_2)=\sum _{i=1}^n\alpha (a_i)\\
   &=\alpha\paren{\sum _{i=1}^na_i}\le\alpha (u_1)=1
\end{align*}
For any $b\in\escript _2$ we have
\begin{equation*}
\iscript ^*(b)=\sqbrac{\sum _{i=1}^n\hscript ^{(a_i,\beta _i)*}}(b)=\sum _{i=1}^n\hscript ^{(a_i,\beta _i)}(b)=\sum _{i=1}^n\beta _i(b)a_i
\end{equation*}
so that
\begin{equation*}
\iscripthat =\iscript ^*(u_2)=\sum _{i=1}^n\beta _i(u_2)a_i=\sum _{i=1}^na_i
\end{equation*}
and $\iscript$ measures the effect $\sum\limits _{i=1}^na_i$. We see that $\iscript$ is a channel if and only if $\sum\limits _{i=1}^na_i=u_2$ in which case
$\brac{a_i\colon i=1,2,\ldots ,n}$ is an observable on $\escript _1$. A mixed Holevo operation $\iscript$ is pure if $\iscript$ can be written $\iscript =\hscript ^{(a,\beta )}$ for some $a\in\escript _1$, $\beta\in\sscript (\escript _2)$.

If $\iscript _i\in\oscript (\escript _1,\escript _2)$ and $\lambda _i\in\sqbrac{0,1}$ with $\sum\limits _{i=1}^n\lambda _i\le 1$, then
$\sum\lambda _i\iscript _i\in\oscript (\escript _1,\escript _2)$ is called a \textit{mixture} of $\iscript _i$, $i=1,2,\ldots ,n$. A mixture of pure Holevo operations is a mixed Holevo operation. Indeed, let $\lambda _i\in\sqbrac{0,1}$ with $\sum _{i=1}^n\lambda _i\le 1$ and let $\iscript =\sum\limits _{i=1}^n\lambda _i\hscript ^{(a_i,\beta _i)}$ be a mixture of pure Holevo operations in $\oscript (\escript _1,\escript _2)$. Then for all $\alpha\in\sscript (\escript _1)$ we have
\begin{align}                
\label{eq41}
\iscript (\alpha )&=\sum _{i=1}^n\lambda _i\hscript ^{(a_i,\beta _i)}(\alpha )=\sum _{i=1}^n\lambda _i\alpha (a_i)\beta _i
   =\sum _{i=1}^n\alpha (\lambda _ia_i)\beta _i\notag\\
   &=\sum _{i=1}^n\hscript ^{(\lambda _ia_i,\beta _i)}(\alpha )
\end{align}
where $\sum\limits _{i=1}^n\lambda _ia_i\le\sum _{i=1}^n\lambda _iu_1\le u_1$. Hence, $\iscript$ is a mixed Holevo operation
$\sum\limits _{i=1}^n\hscript ^{(\lambda _ia_i,\beta _i)}$. We do not know if a mixed Holevo operation is necessarily a mixture of pure Holevo operations.

\begin{example}  
Let $\lambda _i\in\sqbrac{0,1}$ with $\sum\limits _{i=1}^n\lambda _i\le 1$ and let $\iscript =\sum _{i=1}^n\lambda _i\hscript ^{(a_i,\beta )}$, $a_i\in\escript _1$,
$\beta\in\sscript (\escript _2)$. Then $\iscript$ is a mixture of pure Holevo operations and hence $\iscript$ is a mixed Holevo operation in $\oscript (\escript _1,\escript _2)$. Applyng \eqref{eq41}, for all $\alpha\in\sscript (\escript _1)$ we have
\begin{equation*}
\iscript (\alpha )=\sqbrac{\sum _{i=1}^n\alpha (\lambda _ia_i)}\beta =\hscript ^{(\sum\lambda _ia_i,\beta )}(\alpha )
\end{equation*}
Hence, $\iscript =\hscript ^{\paren{\sum\lambda _ia_i,\beta}}$ so $\iscript$ is a pure Holevo operation. Similarly,
$\jscript =\sum\limits _{i=1}^n\lambda _i\hscript ^{(a,\beta _i)}$ is a mixed Holevo operation on $\oscript (\escript _1,\escript _2)$. Applying \eqref{eq41}, for all
$\alpha\in\sscript (\escript _1)$ we have
\begin{equation*}
\jscript (\alpha )=\alpha (a)\sum _{i=1}^n\lambda _i\beta _i=\hscript ^{(a,\sum\lambda _i\beta _i)}(\alpha )
\end{equation*}
Hence, $\jscript =\hscript ^{(a,\sum\lambda _i\beta _i)}$ so $\jscript$ is a pure Holevo operation.\hfill\qedsymbol
\end{example}

\begin{example}  
We give an example of a mixed Holevo operation that is not pure. Let $a\in\escript$ and suppose $\beta _1,\beta _2\in\sscript (\escript )$ satisfy $\beta _1(a)=1$,
$\beta _2(a')=1$. (Such states and effects exist in $\escript (H)$.) It follows that $\beta _1(a')=0$, $\beta _2(a)=0$. Let $\lambda\in (0,1)$ and suppose the mixed Holevo operation
\begin{equation*}
\iscript =\lambda\hscript ^{(a,\beta _1)}+(1-\lambda )\hscript ^{(a',\beta _2)}
\end{equation*}
is pure so there exists $b\in\escript$, $\beta\in\sscript (\escript )$ such that $\iscript =\hscript ^{(b,\beta )}$. We then have
\begin{equation*}
\lambda\alpha (a)\beta _1+(1-\lambda )\alpha (a')\beta _2=\alpha (b)\beta
\end{equation*}
for all $\alpha\in\sscript (\escript )$. Letting $\alpha =\beta _1$ we obtain $\lambda\beta _1=\beta _1(b)\beta$ and letting $\alpha =\beta _2$ gives
$(1-\lambda )\beta _2=\beta _2(b)\beta$. Hence, $\beta _1(b)\ne 0$, $\beta _2(b)\ne 0$ and we have
\begin{equation*}
\frac{\lambda}{\beta _1(b)}\,\beta _1=\frac{(1-\lambda )}{\beta _2(b)}\,\beta _2
\end{equation*}
Operating on $u$ gives $\lambda\beta _2(b)=(1-\lambda )\beta _1(b)$ so $\beta _1=\beta _2$ which results in a contradiction.\hfill\qedsymbol
\end{example}

If $a\in\escript _1$, $b\in\escript _2$, $\beta\in\sscript (\escript _2)$, then $\hscript ^{(a,\beta )}$ measures $a$ and the sequential product relative to $\hscript ^{(a,\beta )}$ becomes
\begin{equation*}
a\sqbrac{\hscript ^{(a,\beta )}}b=\hscript ^{(a,\beta )*}(b)=\beta (b)a
\end{equation*}
For shortness, we write $a\sqbrac{\beta}b=\beta (b)a$ and this is a product that depends on the state $\beta$. More generally, if $a=\sum a_i$, the sequential product of $a$ then $b$ relative to the mixed operation $\sum\hscript ^{(a_i,\beta _i)}$ becomes 
\begin{equation*}
a\sqbrac{\sum _i\hscript ^{(a_i,\beta _i)}}b=\sum _i\hscript ^{(a_i,\beta _i)*}(b)=\sum _i\beta _i(b)a_i=\sum _ia_i\sqbrac{\beta _i}b
\end{equation*}

A \textit{pure Holevo instrument} from $\escript _1$ to $\escript _2$ is a finite set \cite{hol82,hol98}
\begin{equation*}
\iscript =\brac{\hscript ^{(A_x,\beta _x)}\colon x\in\Omega _\iscript}
\end{equation*}
where $A=\brac{A_x\colon x\in\Omega _\iscript}\in\rmob (\escript _1)$, $\beta _x\in\sscript (\escript _2)$. We write
$\iscript _x=\hscript _x^{(A,\beta )}=\hscript ^{(A_x,\beta _x)}$ so
\begin{equation*}
\iscript _x(\alpha )=\hscript ^{(A_x,\beta _x)}(\alpha )=\alpha (A_x)\beta _x
\end{equation*}
is a pure Holevo operation. The corresponding operation-valued measure becomes $\iscript (\Delta )=\sum\limits _{x\in\Delta}\iscript _x$,
$\Delta\subseteq\Omega _\iscript$ and the $\alpha$-distribution is 
\begin{equation*}
\Phi _\alpha ^\iscript (\Delta )=\sum _{x\in\Delta}\iscript _x(\alpha )(u_2)=\sum _{x\in\Delta}\alpha (A_x)
\end{equation*}
for all $\alpha\in\sscript (\escript _1)$ so $\iscripthat _x=A_x$ and $\iscripthat =A$. The dual instrument $\hscript ^{(A,\beta )*}\colon\escript\ _1\to\escript _2$
satisfies
\begin{equation*}
\hscript _x^{(A,\beta )*}(b)=\hscript ^{(A_x,\beta _x)*}(b)=\beta _x(b)A_x
\end{equation*}
for all $x\in\Delta _{\hscript ^{(A,\beta )}}$, $b\in\escript _2$.

For $A,B\in\oscript (\escript _1)$, the \textit{Holevo sequential product of $A$ then} $B$ is the bi-observable in $\rmob (\escript _1)$ given by
\begin{equation*}
(A\circ B)_{xy}=\paren{A\sqbrac{\hscript ^{(A,\beta )}}B}_{xy}=\hscript _x^{(A,\beta )*}(B_y)=\beta _x(B_y)A_x
\end{equation*}
We then have the marginals
\begin{align*}
(A\circ B)_x^1&=\sum _y\beta _x(B_y)A_x=A_x\\
(A\circ B)_y^2&=\sum _x\beta _x(B_y)A_x=\sum _xA_x\sqbrac{\beta _x}B_y=\paren{B\ab{\hscript ^{(A,\beta )}}A}
\end{align*}
If $\hscript ^{(A,\alpha )}\in\rmin (\escript _1,\escript _2)$, $\hscript ^{(B,\beta )}\in\rmin (\escript _2,\escript _3)$ are pure Holevo instruments, then the sequential product in
$\rmin (\escript _1,\escript _3)$ becomes
\begin{align*}
(\hscript ^{(A,\alpha )}\circ\hscript ^{(B,\beta )})_{xy}(\gamma )&=(\hscript _x^{(A,\alpha )}\circ\hscript _y^{(B,\beta )})(\gamma )
   =\hscript _y^{(B,\beta )}\sqbrac{\hscript _x^{(A,\alpha )}(\gamma )}\\
   &=\hscript _y^{(B\,\beta )}\sqbrac{\gamma (A_x)\alpha _x}=\gamma (A_x)\hscript _y^{(B,\beta )}(\alpha _x)\\
   &=\gamma (A_x)\alpha _x(B_y)\beta _y=\gamma\sqbrac{\alpha _x(B_y)A_x}\beta _y\\
   &=\gamma\sqbrac{(A\circ B)_{xy}}\beta _y=\hscript _{xy}^{(A\circ B,\beta ')}(\gamma )
\end{align*}
for all $\gamma\in\sscript (\escript _1)$, where $\beta '_{xy}=\beta _y$ for all $x\in\Omega _A$, $y\in\Omega _B$. Therefore,
\begin{equation*}
\hscript ^{(A,\alpha )}\circ\hscript ^{(B,\beta )}=\hscript ^{(A\circ B,\beta ')}
\end{equation*}
is a pure Holevo bi-instrument that measures $A\circ B$ and has states $\beta '_{xy}=\beta _y$. The marginals become
\begin{align*}
(\hscript ^{(A\circ B,\beta ')1})_x(\gamma )&=\sum _y\gamma\sqbrac{(A\circ B)_{xy}}\beta '_{xy}=\sum _y\gamma (A_x)\alpha _x(B_y)\beta _y\\
   &=\gamma (A_x)\sum _y\alpha _x(B_y)\beta _y=\gamma (A_x)\delta _x=\hscript _x^{(A,\delta )}(\gamma )
\end{align*}
for all $\gamma\in\sscript (\escript _1)$ where $\delta _x=\sum _y\alpha _x(B_y)\beta _y\in\sscript (\escript _3)$. Hence, $\hscript ^{(A\circ B,\beta ')}=\hscript ^{(A,\delta )}$.
Moreover,
\begin{align*}
(\hscript ^{(A\circ B,\beta ')2})_y(\gamma )&=\sum _x\gamma\sqbrac{(A\circ B)_{xy}}\beta '_{xy}=\sum _x\gamma (A_x)\alpha _x(B_y)\beta _y\\
   &=\gamma\sqbrac{\sum _x\alpha _x(B_y)A_x}\beta _y=\hscript _y^{(C,\beta )}(\gamma )
\end{align*}
for every $\gamma\in\sscript (\escript _1)$ where $C\in\rmob (\escript _1)$ is the observable
\begin{equation*}
C_y=\sum _x\alpha _x(B_y)A_x=(A\circ B)^2(y)=\sum _xA_x\sqbrac{\alpha _x}B_y
\end{equation*}
Hence,
\begin{equation*}
\hscript ^{(A\circ B,\beta ')2}=\hscript ^{(C,\beta )}=\hscript ^{\paren{(A\circ B)^2,\beta}}
\end{equation*}
We conclude that $\hscript ^{(A,\delta )}$ and $\hscript ^{\paren{(A\circ B)^2,\beta}}$ coexist with joint instrument $\hscript ^{(A\circ B,\beta ')}$. Moreover, 
$\hscript ^{(A\circ B,\beta ')1}$ measures $A$ and $\hscript ^{(A\circ B,\beta ')2}$ measures $(A\circ B)^2$.

A \textit{mixed Holevo instrument} has the form $\iscript =\sum\lambda _i\hscript ^{(A_i,\beta _i)}$ where $A_i\in\rmob (\escript _1)$, $\beta _{ix}\in\sscript (\escript _2)$,
$\Omega _{A_i}=\Omega$, $i=1,2,\ldots ,n$, $\lambda _i\in\sqbrac{0,1}$, $\sum\lambda _i=1$. It is not completely clear that $\iscript$ is indeed an instrument and this is shown in the following theorem.

\begin{thm}    
\label{thm41}
{\rm{(i)}}\enspace $\iscript =\sum\lambda _i\hscript ^{(A_i,\beta _i)}\in\rmin (\escript _1,\escript _2)$.
{\rm{(ii)}}\enspace For $\alpha\in\sscript (\escript _1)$, letting $\alpha _i=\sum\limits _x\alpha (A_{ix})\beta _{ix}$ we have that $\alpha _i\in\sscript (\escript _2)$ and
$\iscriptbar (\alpha )=\sum\lambda _i\alpha _i\in\sscript (\escript _2)$.
{\rm{(iii)}}\enspace $\iscripthat =\sum\lambda _iA_i$.
\end{thm}
\begin{proof}
For all $x\in\Omega$, $\alpha\in\sscript (\escript _1)$ we obtain
\begin{equation*}
\iscript _x(\alpha )=\sum _i\lambda _i\hscript _x^{(A_i,\beta _i)}(\alpha )=\sum _i\lambda _i\alpha (A_{ix})\beta _{ix}
\end{equation*}
Since $\alpha (A_{ix})\le 1$, we have $\sum\limits _i\lambda _i\alpha (A_{ix})\le\sum\limits _x\lambda _i=1$. Hence,
$\iscript\colon\sscript (\escript )\to\rmsub\paren{\sscript (\escript _2)}$ and it is clear that $\iscript$ is affine. Since $\sum \limits _x\alpha (A_{ix})=1$, $i=1,2,\ldots ,n$, we have $\alpha _i\in\sscript (\escript _2)$. Hence,
\begin{align*}
\iscriptbar (\alpha )&=\sum _x\iscript _x(\alpha )=\sum _x\sum _i\lambda _i\alpha (A_{ix})\beta _{Ix}=\sum _i\lambda _i\sqbrac{\sum _x\alpha (A_{ix})\beta _{ix}}\\
   &=\sum _i\lambda _i\alpha _i\in\sscript (\escript _2)
\end{align*}
This proves (ii) and shows that $\iscript _x$ is a mixed Holevo operation and $\iscriptbar$ is a channel. It follows that $\iscript\in\rmin (\escript _1,\escript _2)$ so (i) holds.
To prove (iii) we have
\begin{align*}
\iscript _x^*(u_2)&=\sum _i\lambda _i\hscript _x^{(A_i,\beta _i)*}(u_2)=\sum _i\lambda _i\hscript ^{(A_{ix},\beta _{ix})*}(u_2)\\
   &=\sum _i\lambda _i\beta _{ix}(u_2)A_{ix}=\sum _i\lambda _iA_{ix}=\paren{\sum _i\lambda _iA_i}_x
\end{align*}
Hence, $\iscripthat =\sum\lambda _iA_i$.
\end{proof}

As with operations, a mixed Holevo instrument need not be pure. 

The next result summarizes properties of the sequential product $a\sqbrac{\beta}b=\beta (b)a$ \cite{gn01,gud22,gud23}.

\begin{lem}    
\label{lem42}
{\rm{(i)}}\enspace $b\mapsto a\sqbrac{\beta}b$ and $a\mapsto a\sqbrac{\beta}b$ are convex and additive.
{\rm{(ii)}}\enspace $\theta\sqbrac{\beta}b=b\sqbrac{\beta}\theta =\theta$,
{\rm{(iii)}}\enspace $a\sqbrac{\beta}u=a$, $u\sqbrac{\beta}b=\beta (b)u$ and $u\sqbrac{\beta}b=b$ if and only if $b=\lambda u$ for $\lambda\in\sqbrac{0,1}$.
{\rm{(iv)}}\enspace $a\sqbrac{\beta}b\le a$,
{\rm{(v)}}\enspace the associative law holds: $a\sqbrac{\alpha}\paren{b\sqbrac{\beta}c}=\paren{a\sqbrac{\beta}b}\sqbrac{\beta}c$,
{\rm{(vi)}}\enspace $a\sqbrac{\alpha}b=\theta\not\Rightarrow b\sqbrac{\alpha}a=\theta$.
\end{lem}
\begin{proof}
The proofs of (i)--(iv) are straightforward. To prove (v) we have
\begin{align*}
a\sqbrac{\alpha}\paren{b\sqbrac{\beta}c}&=a\sqbrac{\alpha}\paren{\beta (c)b}=\beta (c)a\sqbrac{\alpha}b=\beta (c)\alpha (b)a\\
   &=\paren{a\sqbrac{\alpha}b}\beta (c)=\paren{a\sqbrac{\alpha}b}\sqbrac{\beta}c
\end{align*}
To prove (vi), let $a=u$, $b\ne\theta$ and $\alpha\in\sscript (\escript )$ with $\alpha (b)=0$. Then 
\begin{equation*}
a\sqbrac{\alpha}b=\alpha (b)a=\theta
\end{equation*}
but
\begin{equation*}
b\sqbrac{\alpha}a=\alpha (a)b=b\ne\theta\qedhere
\end{equation*}
\end{proof}

For $\alpha\in\sscript (\escript )$, $a,b\in\escript$, the $\alpha$-\textit{commutant} of $a,b$ is 
\begin{equation*}
\sqbrac{a,b}_\alpha=a\sqbrac{\alpha}b-b\sqbrac{\alpha}a=\alpha (b)a-\alpha (a)b
\end{equation*}
If $\alpha (a),\alpha (b)\ne 0$ we see that $\sqbrac{a,b}_\alpha =0$ if and only if $a=\lambda b$ for some $\lambda\ge 0$. It is also clear that
$a\mapsto\sqbrac{a,b}_\alpha$ is additive and convex. Moreover, 
\begin{equation*}
\sqbrac{\lambda a,b}_\alpha =\sqbrac{a,\lambda b}_\alpha =\lambda\sqbrac{a,b}_\alpha
\end{equation*}
for all $\lambda\in\sqbrac{0,1}$. We also have: $\sqbrac{a,b}_\alpha =-\sqbrac{b,a}_\alpha$, $\alpha\paren{\sqbrac{a,b}_\alpha}=0$, $\sqbrac{a,a}_\alpha =0$.
The next result gives other properties of $\sqbrac{a,b}_\alpha$.

\begin{lem}    
\label{lem43}
{\rm{(i)}}\enspace If $\sqbrac{a,c}_\alpha =\theta$, then $\sqbrac{a\sqbrac{\alpha}b,c}=\theta$ for all $b\in\escript$.
{\rm{(ii)}}\enspace $\sqbrac{a+b,c}_\alpha=\sqbrac{a,c}_\alpha +\sqbrac{b,c}_\alpha$.
{\rm{(iii)}}\enspace In general $\sqbrac{a,u}_\alpha\ne\theta$.
{\rm{(iv)}}\enspace $\sqbrac{a,b}_\alpha =\theta\not\Rightarrow\sqbrac{a,b'}_\alpha =\theta$.
\end{lem}
\begin{proof}
(i)\enspace If $\sqbrac{a,c}_\alpha =\theta$, then $\alpha (c)a=\alpha (a)c$. Hence,
\begin{align*}
\sqbrac{a\sqbrac{\alpha}b,c}_\alpha&=\sqbrac{\alpha (b)a,c}_\alpha =\alpha (c)\alpha (b)a-\alpha (b)\alpha (a)c\\
   &=\alpha (b)\alpha (a)c-\alpha (b)\alpha (a)c=\theta
\end{align*}
(ii)\enspace This follows from
\begin{align*}
\sqbrac{a+b,c}_\alpha&=\alpha (c)(a+b)-\alpha (a+b)c=\alpha (c)a-\alpha (a)c+\alpha (c)b-\alpha (b)c\\
   &=\sqbrac{a,c}_\alpha +\sqbrac{b,c}_\alpha
\end{align*}
(iii)\enspace If $\alpha (a)u\ne a$, we obtain
\begin{equation*}
\sqbrac{a,u}_\alpha =\alpha (u)a-\alpha (a)u=a-\alpha (a)u\ne\theta
\end{equation*}
(iv)\enspace Letting $b=\theta$ and $a\ne\alpha (a)u$ we have $\sqbrac{a,\theta}_\alpha =\theta$ but
\begin{equation*}
\sqbrac{a,b'}_\alpha =\sqbrac{a,u}_\alpha =\alpha (u)a-\alpha (u)a-\alpha (a)u=a-\alpha (a)u\ne\theta\qedhere
\end{equation*}
\end{proof}

\end{document}